\documentclass[a4paper]{article}

\usepackage[algoruled,linesnumbered,algo2e,vlined]{algorithm2e}
\usepackage{amsmath,amssymb}
\usepackage{url}
\usepackage{rotating}
\usepackage{comment}
\usepackage{multirow}
\usepackage{amsthm}
\theoremstyle{plain}
\newtheorem{theorem}{Theorem}

\newtheorem{lemma}[theorem]{Lemma}

\newcommand{\emptystr}{\varepsilon}

\newcommand{\depth}{\mathsf{d}}
\newcommand{\str}{\mathsf{str}}
\newcommand{\strs}{\mathsf{s}}
\newcommand{\strt}{\mathsf{t}}
\newcommand{\Nodes}{\mathit{V}}
\newcommand{\NNodes}{\hat{\mathit{V}}}
\newcommand{\DMN}{\mathit{D}}
\newcommand{\DMNRoot}{\mathit{D}_{\mathsf{root}}}

\newcommand{\gtext}{\mathit{T}}

\newcommand{\rST}{\mathcal{T}}
\newcommand{\Substr}{\mathit{S}} 
\newcommand{\ddiff}{\alpha} 
\newcommand{\LCP}{\mathsf{lcp}} 
\newcommand{\ramsort}{\mathit{C}_{\mathsf{sort}}} 
\newcommand{\bnode}{\mathsf{b}}
\newcommand{\rfr}{\mathsf{rem}}

\newcommand{\idtt}[1]{\ensuremath{\mathtt{#1}}}

\title{Substring Complexities on Run-length Compressed Strings}

\author{
  Akiyoshi~Kawamoto\\
  {Kyushu Institute of Technology, Japan}\\
  {\texttt{kawamoto.akiyoshi256@mail.kyutech.jp}}\\
  \\
  Tomohiro~I\\
  {Kyushu Institute of Technology, Japan}\\
  {\texttt{tomohiro@ai.kyutech.ac.jp}}\\
}

\date{}
\begin{document}

\maketitle

\begin{abstract}
Let $\Substr_{T}(k)$ denote the set of distinct substrings of length $k$ in a string $T$,
then the $k$-th substring complexity is defined by its cardinality $|\Substr_{T}(k)|$.
Recently, $\delta = \max \{ |\Substr_{T}(k)| / k : k \ge 1 \}$ is shown to be a good compressibility measure of highly-repetitive strings.
In this paper, given $T$ of length $n$ in the run-length compressed form of size $r$,
we show that $\delta$ can be computed in $\mathit{C}_{\mathsf{sort}}(r, n)$ time and $O(r)$ space,
where $\mathit{C}_{\mathsf{sort}}(r, n) = O(\min (r \lg\lg r, r \lg_{r} n))$ is the time complexity 
for sorting $r$ $O(\lg n)$-bit integers in $O(r)$ space in the Word-RAM model with word size $\Omega(\lg n)$.
\end{abstract}

\section{Introduction}\label{sec:intro}
Data compression has been one of the central topics in computer science and recently
becomes much more important as data continues growing faster than ever before.
One data category that is rapidly increasing is \emph{highly-repetitive strings}, which have many long common substrings.
Typical examples of highly-repetitive strings are genomic sequences collected from similar species and versioned documents.
It is known that highly-repetitive strings can be compressed much smaller than entropy-based compressors
by utilizing repetitive-aware data compressions such as LZ76~\cite{1976LempelZ_ComplOfFinitSequen_TIT}, 
bidirectional macro scheme~\cite{1982StorerS_DataComprViaTexturSubst},
grammar compression~\cite{Kieffer2000GrammarBasedCodes},
collage system~\cite{Kida2003Csu}, and
run-length compression of Burrows-Wheeler Transform~\cite{Burrows1994BWT}.
Moreover, many efforts have been devoted for efficient algorithms to conduct useful 
operations on compressed data (without explicitly decompressing it) such as random access and string pattern matching.
We refer readers to the comprehensive survey~\cite{2021Navarro_IndexHighlRepetStrinCollec_I,2021Navarro_IndexHighlRepetStrinCollec_II} for recent developments in this area.

Since the compressibility of highly-repetitive strings is not captured by the information entropy,
how to measure compressibility for highly-repetitive strings is a long standing question.
Addressing this problem, Kempa and Prezza~\cite{2018KempaP_AtRootsOfDictionCompr_STOC} proposed a new concept of string attractors;
a set of positions $\Gamma$ is called an attractor of a string $\gtext$ iff
any substring of $\gtext$ has at least one occurrence intersecting with a position in $\Gamma$.
They showed that there is a string attractor behind existing repetitive-aware data compressions and 
the size $\gamma$ of the smallest attractor lower bounds and well approximates the size of them.
In addition, this and subsequent studies revealed that it is possible to design ``universal'' compressed data structures
built upon any string attractor $\Gamma$ to support operations such as random access and string pattern matching~\cite{2018KempaP_AtRootsOfDictionCompr_STOC,2019Prezza_OptimRankAndSelecQueries_CPM,2019NavarroP_UniverComprTextIndex}.
A drawback of string attractors is that computing the size $\gamma$ of the smallest string attractors is NP-hard~\cite{2018KempaP_AtRootsOfDictionCompr_STOC}.
As a substitution of $\gamma$, Christiansen et al.~\cite{2021ChristiansenEKNP_OptimTimeDictionComprIndex} proposed another 
repetitive-aware compressibility measure $\delta$ defined as the maximum of normalized substring complexities $\max \{ |\Substr_{\gtext}(k)| / k : k \ge 1 \}$,
where $\Substr_{\gtext}(k)$ denotes the set of substrings of length $k$ in $\gtext$.
They showed that $\delta \le \gamma$ always holds and there is an $O(\delta \lg (n / \delta))$-size data structure for 
supporting efficient random access and string pattern matching.

Although substring complexities were used in~\cite{2013RaskhodnikovaRRS_SublinAlgorForApproxStrin} to approximate the number of LZ76 phrases in sublinear time,
it is very recent that $\delta$ gains an attention as a gold standard of repetitive-aware compressibility measures~\cite{2021ChristiansenEKNP_OptimTimeDictionComprIndex}.
Recent studies have shown that $\delta$ possesses desirable properties as a compressibility measure.
One of the properties is the robustness.
For example, we would hope for a compressibility measure to monotonically increase while appending/prepending characters to a string.
It was shown that $\delta$ has this monotonicity~\cite{2020KociumakaNP_TowarDefinMeasurOfRepet_LATIN} 
while $\gamma$ is not~\cite{2021MantaciRRRS_CombinViewOnStrinAttrac}.
Akagi et al.~\cite{2021AkagiFI_SensitOfStrinComprAnd_X} studied the sensitivity of repetitive measures 
to edit operations and the results confirmed the robustness of $\delta$.

From algorithmic aspects, for a string of length $n$, $\delta$ can be computed in $O(n)$ time and space~\cite{2021ChristiansenEKNP_OptimTimeDictionComprIndex}.
This contrasts with $\gamma$ and the smallest compression size in some compression schemes 
like bidirectional macro scheme, grammar compression and collage system, which are known to be NP-hard to compute.
Exploring time-space tradeoffs, Bernardini et al.~\cite{2020BernardiniFGP_SubstComplInSublinSpace_X} presented 
an algorithm to compute $\delta$ with sublinear additional working space.
Efficient computation of $\delta$ has a practical importance for 
the $O(\delta \lg (n / \delta))$-size data structure of~\cite{2021ChristiansenEKNP_OptimTimeDictionComprIndex} 
as once we know $\delta$, we can built the data structure 
while reading $\gtext$ $O(\lg (n / \gamma))$ expected times in a streaming manner with main-memory space bounded in terms of $\delta$.

In this paper, we show that $\delta$ can be computed in $\ramsort(r, n)$ time and $O(r)$ space,
where $\ramsort(r, n)$ is the time complexity 
for sorting $r$ $O(\lg n)$-bit integers in $O(r)$ space in the Word-RAM model with word size $\Omega(\lg n)$.
We can easily obtain $\ramsort(r, n) = O(r \lg r)$ by comparison sort and 
$\ramsort(r, n) = O(r \lg_r n)$ by radix sort that uses $\Theta(r)$-size buckets.
Plugging in more advanced sorting algorithms for word-RAM model, $\ramsort(r, n) = O(r \lg \lg r)$~\cite{2004Han_DeterSortinInOEmphn}.
In randomized setting, 
$\ramsort(r, n) = O(r)$ if $\lg n = \Omega(\lg^{2+\epsilon} r)$ for some fixed $\epsilon > 0$~\cite{1998AnderssonHNR_SortinInLinearTime},
and otherwise $\ramsort(r, n) = O(r \sqrt{\lg \frac{\lg n}{\lg r}})$~\cite{2002HanT_IntegSortinIn0N_FOCS}.

\section{Preliminaries}\label{sec:prelim}

Let $\Sigma$ be a finite {\em alphabet}.
An element of $\Sigma^*$ is called a {\em string} over $\Sigma$.
The length of a string $w$ is denoted by $|w|$. 
The empty string $\emptystr$ is the string of length 0,
that is, $|\emptystr| = 0$.
Let $\Sigma^+ = \Sigma^* - \{\emptystr\}$.
The concatenation of two strings $x$ and $y$ is denoted by $x \cdot y$ or simply $xy$.
When a string $w$ is represented by the concatenation of strings $x$, $y$ and $z$ (i.e. $w = xyz$), 
then $x$, $y$ and $z$ are called a \emph{prefix}, \emph{substring}, and \emph{suffix} of $w$, respectively.
The $i$-th character of a string $w$ is denoted by $w[i]$ for $1 \leq i \leq |w|$,
and the substring of a string $w$ that begins at position $i$ and
ends at position $j$ is denoted by $w[i..j]$ for $1 \leq i \leq j \leq |w|$,
i.e., $w[i..j] = w[i]w[i+1] \dots w[j]$.
For convenience, let $w[i..j] = \emptystr$ if $j < i$.
For two strings $x$ and $y$, let $\LCP(x, y)$ denote the length of the longest common prefix between $x$ and $y$.
For a character $c$ and integer $e \ge 0$, let $c^e$ denote the string consisting of a single character $c$ repeated by $e$ times.

A substring $w[i..j]$ is called a \emph{run} of $w$ iff it is a maximal repeat of a single character,
i.e., $w[i-1] \neq w[i] = w[i+1] = \dots = w[j-1] = w[j] \neq w[j+1]$ (the inequations for $w[i-1]$ and/or $w[j+1]$ are ignored if they are out of boundaries).
We obtain the run-length encoding of a string by representing each run by a pair of character and the number of repeats in $O(1)$ words of space.
For example, a string \idtt{aabbbaabb} has four runs \idtt{aa}, \idtt{bbb}, \idtt{aa} and \idtt{bb},
and its run-length encoding is $(\idtt{a}, 2), (\idtt{b}, 3), (\idtt{a}, 2), (\idtt{b}, 2)$ or we just write as $\idtt{a}^2\idtt{b}^3\idtt{a}^2\idtt{b}^2$.

Throughout this paper, we refer to $\gtext$ as a string over $\Sigma$ of length $n$ with 
$r$ runs and $\delta = \max \{|\Substr_{\gtext}(k)| / k : k \ge 1\}$.
Our assumption on computational model is the Word-RAM model with word size $\Omega(\lg n)$.
We assume that $\Sigma$ is interpreted as an integer with $O(1)$ words or $O(\lg n)$ bits,
and the order of two characters in $\Sigma$ can be determined in constant time.

Let $\rST$ denote the trie representing the suffixes of $\gtext$ that start with the run's boundaries in $\gtext$, 
which we call the \emph{r-suffix trie} of $\gtext$.
Note that $\rST$ has at most $r$ leaves and at most $r$ branching internal nodes, and thus, 
$\rST$ can be represented in $O(r)$ space by compacting non-branching internal nodes
and representing edge labels by the pointers to the corresponding fragments in the run-length encoded string of $\gtext$.
We call the compacted r-suffix trie the \emph{r-suffix tree}.
In order to work in $O(r)$ space, our algorithm actually works on the r-suffix tree of $\gtext$,
but our conceptual description will be made on the r-suffix trie considering that non-branching internal nodes remain.
A node is called \emph{explicit} when we want to emphasize that the node is present in the r-suffix tree, too.

Let $\Nodes$ denote the set of nodes of $\rST$.
For any $v \in \Nodes$, let $\str(v)$ be the string obtained by concatenating the edge labels from the root to $v$,
and $\depth(v) = |\str(v)|$.
We say that $v$ represents the string $\str(v)$ and sometimes identify $v$ by $\str(v)$ if it is clear from the context.
For any $v \in \Nodes$ and $1 \le k \le \depth(v)$, let $\str_k(v)$ be the suffix of $\str(v)$ of length $k$.
Let $\NNodes$ denote the set of nodes that do not have a child whose string consists of a single character.
Figure~\ref{fig:r_suffix_trie} shows an example of $\rST$ for $\gtext = \idtt{aabbbaabbaaa}$.

\begin{figure}[t]
\begin{center}
  \includegraphics[scale=0.57]{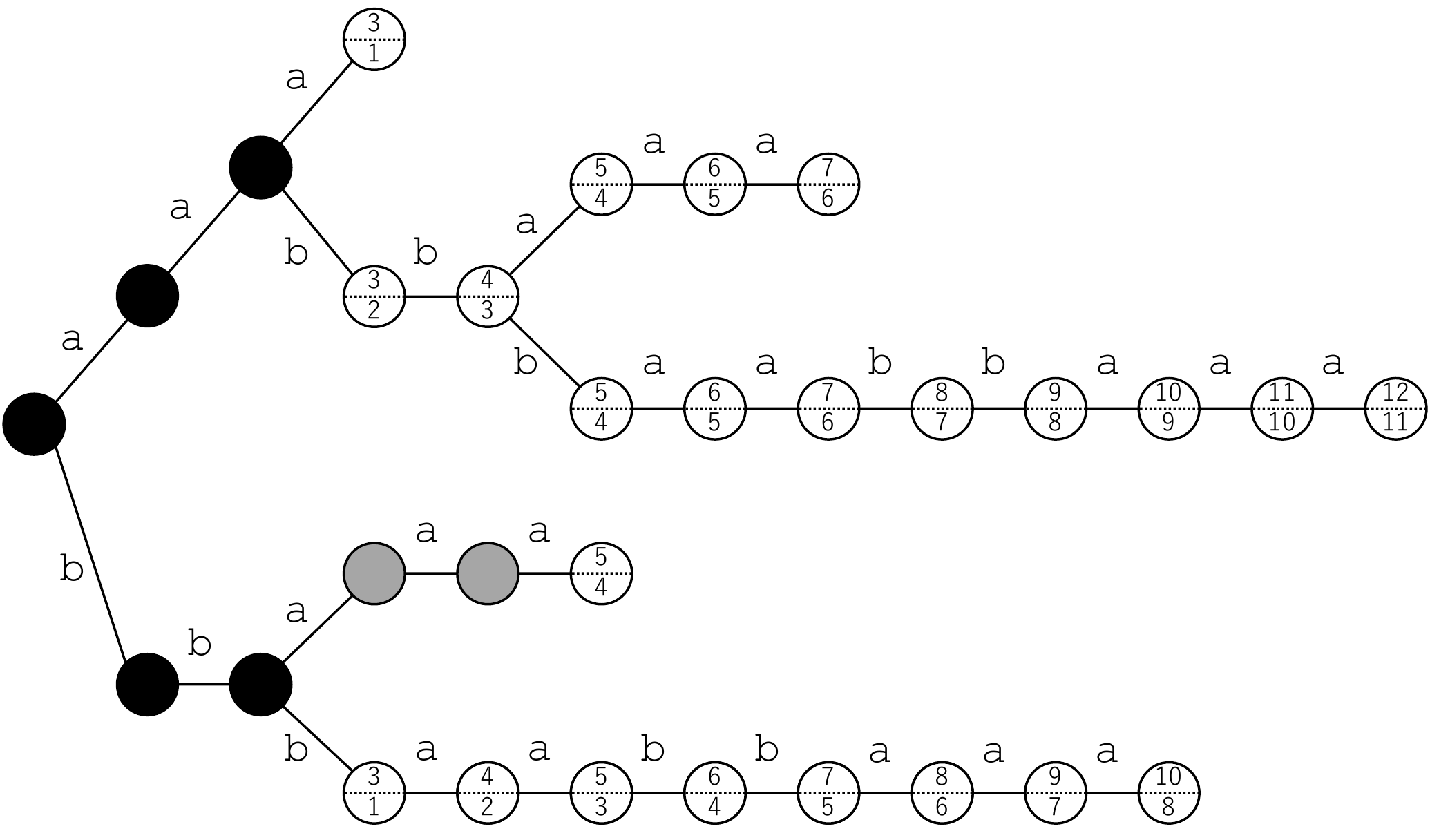}
  \caption{
    An example of r-suffix trie $\rST$ for $\gtext = \idtt{aabbbaabbaaa}$.
    The nodes in black $\{\emptystr, \idtt{a}, \idtt{aa}, \idtt{b}, \idtt{bb} \}$ are the nodes in $\Nodes - \NNodes$.
    With the notations introduced in Section~\ref{sec:properties},
    the nodes in gray \idtt{bba} and \idtt{bbaa} cannot be deepest matching nodes due to \idtt{bbba} and \idtt{bbbaa}, respectively.
    The other nodes belong to $\DMN$ and two integers in node $v \in \DMN$ represent $\depth(v)$ (upper integer) and $|\strt(v)|$ (lower integer), 
    where $\strt(v)$ is defined in Section~\ref{sec:properties}.
  }
\label{fig:r_suffix_trie}
\end{center}
\end{figure}

\section{Connection between $\Substr_{\gtext}(k)$ and the r-suffix trie}\label{sec:properties}
We first recall a basic connection between $\Substr_{\gtext}(k)$ and the nodes of the \emph{suffix trie} of $\gtext$
that is the trie representing ``all'' the suffixes of $\gtext$~\cite{1973Weiner_LinearPatterMatchAlgor,GusfieldBook}.
Since a $k$-length substring $w$ in $\gtext$ is a $k$-length prefix of some suffix of $\gtext$,
$w$ can be uniquely associated with the node $v$ representing $w$ by the path from the root to $v$.
Thus, the set $\Substr_{\gtext}(k)$ of substrings of length $k$ is captured by the nodes of depth $k$.
This connection is the basis of the algorithm presented in~\cite[Lemma 5.7]{2021ChristiansenEKNP_OptimTimeDictionComprIndex}
to compute $\delta$ in $O(n)$ time and space.

We want to establish a similar connection between $\Substr_{\gtext}(k)$ and the nodes of the r-suffix trie $\rST$ of $\gtext$.
Since only the suffixes that start with run's boundaries are present in $\rST$,
there could be a substring $w$ of $\gtext$ that is not represented by a path from the root to some node.
Still, we can find an occurrence of any non-empty substring $w$ in a path starting from a node in $\Nodes - \NNodes$:
Suppose that $w = \gtext[i..j]$ is an occurrence of $w$ in $\gtext$ and $i' \le i$ is the starting position of the run containing $i$, then
there is a node $v$ such that $\str(v) = \gtext[i'..j]$ and $\str_{|w|}(v) = w$.
Formally, we say that node $v$ is a \emph{matching node} for a string $w$ iff $\str(v) = w[1]^e w$ for some integer $e \ge 0$.
Note that there could be more than one matching nodes for a string $w$, but
the \emph{deepest matching node} for $w$ is unique because 
two distinct matching nodes must have different values of $e$.
The following lemma summarizes the above discussion.
\begin{lemma}\label{lem:dmn}
  For any substring $w$ of $\gtext$, there is a unique deepest matching node in $\rST$.
\end{lemma}

The next lemma implies that a deepest matching node for some string is in $\NNodes$.
\begin{lemma}\label{lem:dmn_in_nnodes}
  A node $v \notin \NNodes$ cannot be a deepest matching node for any string.
\end{lemma}
\begin{proof}
  Since $v \notin \NNodes$, there exists a child $v'$ of $v$ such that $\str(v') = c^{e}$ for some character $c$ and integer $e > 0$.
  It also implies that $\str(v) = c^{e-1}$.
  Thus, $v'$ is always a deeper matching node for any suffix of $\str(v)$ and the claim of the lemma follows.
\end{proof}

Let $\DMN \subseteq \NNodes$ denote the set of deepest matching nodes for some strings,
and $\DMN_{k} \subseteq \DMN$ denote the set of deepest matching nodes for some strings of length $k$.
For fixed $k$, it is obvious that a node $v$ in $\DMN_{k}$ is the deepest matching node for a unique substring of length $k$, which is $\str_{k}(v)$.
Together with Lemma~\ref{lem:dmn}, there is a bijection between $\Substr_{\gtext}(k)$ and $\DMN_{k}$, which leads to the following lemma.
\begin{lemma}\label{lem:bijection}
  For any $1 \le k \le n$, $|\Substr_{\gtext}(k)| = |\DMN_{k}|$.
\end{lemma}
By Lemma~\ref{lem:bijection}, $\delta$ can be computed by $\max\{ |\DMN_{k}| / k : k \ge 1\}$.

\begin{figure}[t]
\begin{center}
  \includegraphics[scale=0.57]{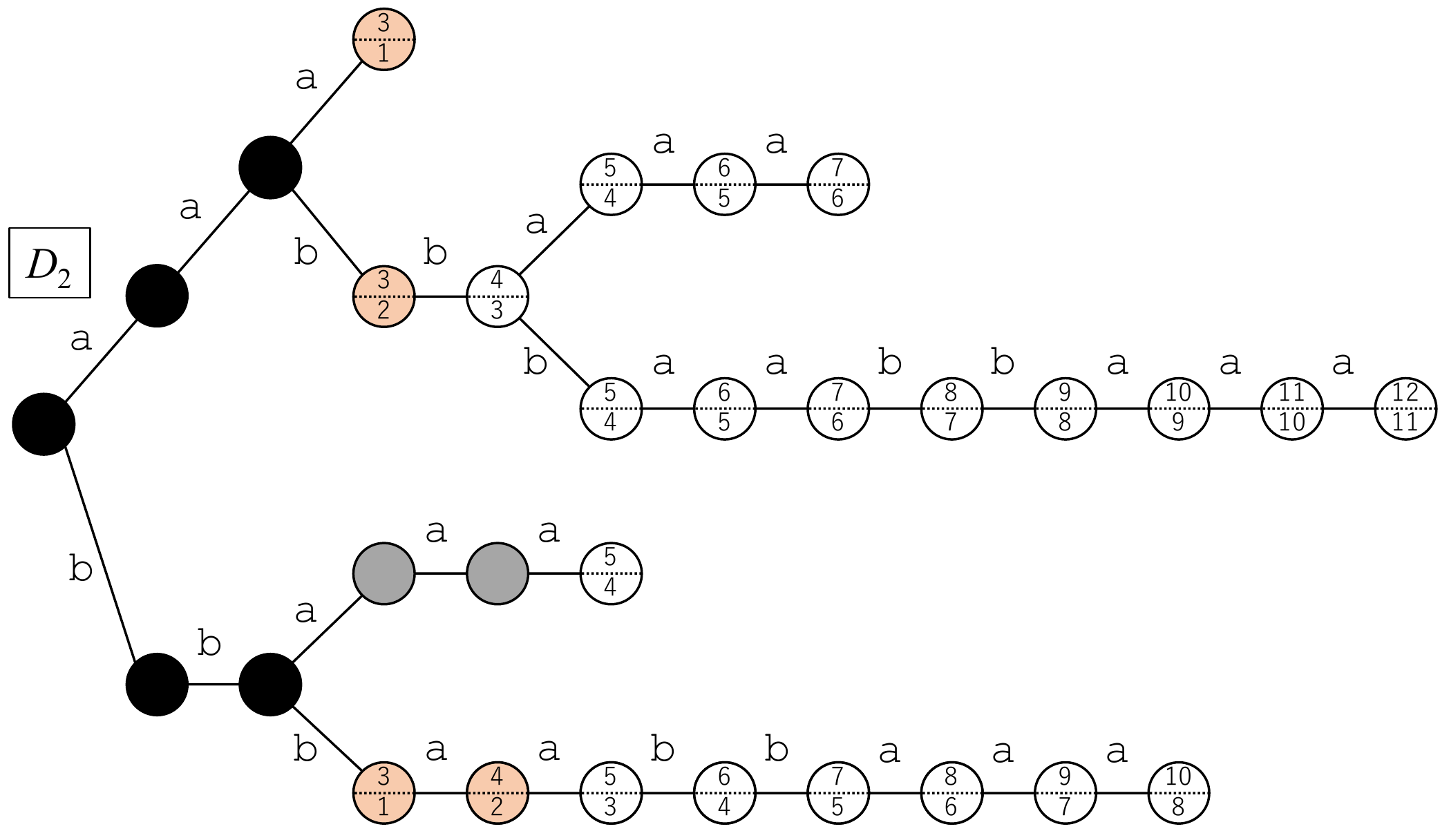}
  \caption{
    Displaying the nodes in $\DMN_{2}$ in orange filled color. 
    The nodes $\idtt{aaa}, \idtt{aab}, \idtt{bbb}$ and $\idtt{bbba}$ are respectively 
    the deepest matching nodes for $\idtt{aa}, \idtt{ab}, \idtt{bb}$ and $\idtt{ba}$.
  }
\label{fig:rst_k4}
\end{center}
\end{figure}

Next we study some properties on $\DMN$, which is used to compute $\delta$ efficiently.
For any $v \in \NNodes$, let $\strs(v)$ denote the string that can be obtained by removing one character from the first run of $\str(v)$,
and let $\strt(v)$ be the string such that $\str(v) = \strs(v) \cdot \strt(v)$.
In other words, $\strt(v)$ is the shortest string for which $v$ is a matching node.
For example, if $\str(v) = \idtt{aabba}$ then $\strs(v) = \idtt{a}$ and $\strt(v) = \idtt{abba}$.
\begin{lemma}\label{lem:dmn_for_v}
  For any $v \in \NNodes$, $\{ k : v \in \DMN_{k} \}$ is $[|\strt(v)|..\depth(v)]$ or $\emptyset$.
\end{lemma}
\begin{proof}
  First of all, by definition it is clear that $v$ cannot be a matching node for a string of length shorter than $|\strt(v)|$ or longer than $\depth(v)$.
  Hence, only the integers in $[|\strt(v)|..\depth(v)]$ can be in $\{ k : v \in \DMN_{k} \}$,
  and the claim of the lemma says that $\{ k : v \in \DMN_{k} \}$ either takes them all or nothing.
  
  We prove the lemma by showing that if $v$ is not the deepest matching node for some suffix $w$ of $\str(v)$ with $|\strt(v)| \le |w| \le \depth(v)$,
  then $v$ is not the deepest matching node for any string.
  The assumption implies that there is a deeper matching node $v'$ for $w$.
  Let $c = \strt(v)[1]$, then $\str(v) = c^e \strt(v)$ with $e = |\strs(v)|$.
  Note that $w$ is obtained by prepending zero or more $c$'s to $\strt(v)$, and hence, $\str(v')$ is written as $c^{e'}w$ for some $e' > e$.
  Therefore, $v$ and $v'$ are matching nodes for any suffix of $\str(v)$ of length in $[|\strt(v)|..\depth(v)]$.
  Since $v'$ is deeper than $v$, the claim holds.
\end{proof}

\begin{lemma}\label{lem:child_is_dmn}
  For any $u \in \DMN$, any child $v$ of $u$ is in $\DMN$.
\end{lemma}
\begin{proof}
  We show a contraposition, i.e., $u \notin \DMN$ if $v \notin \DMN$.
  Let $c = \strt(v)[1]$, then $\str(v) = c^e \strt(v)$ with $e = |\strs(v)|$.
  Note that, by Lemma~\ref{lem:dmn_for_v}, $v \in \DMN$ iff $v$ is the deepest matching node for $\strt(v)$.
  We assume that $u \in \NNodes$ since otherwise $u \notin \DMN$ is clear.
  Then, it holds that $u \in \DMN$ iff $u$ is the deepest matching node for $\strt(v)[1..|\strt(v)|-1]$.
  The assumption of $v \notin \DMN$ then implies that there is a deeper matching node $v'$ for $\strt(v)$ such that $\str(v') = c^{e'} \strt(v)$ with $e' > e$.
  Since the parent $u'$ of $v'$ is deeper than $u$ and a matching node for $\strt(v)[1..|\strt(v)|-1]$,
  we conclude that $u \notin \DMN$.
\end{proof}

By Lemma~\ref{lem:child_is_dmn}, we can identify some deepest matching node $v$ such that
all the ancestors of $v$ are not in $\DMN$ and all the descendants of $v$ are in $\DMN$.
We call such a node $v$ a \emph{DMN-root}, and let $\DMNRoot$ denote the set of DMN-roots.
We will later in Section~\ref{sec:algorithm} show how to compute $\DMNRoot$ in $\ramsort(r, n)$ time and $O(r)$ space.

Figure~\ref{fig:rst_k} shows how $\DMN_{k}$ changes when we increase $k$.
\begin{figure}[t]
\begin{center}
  \includegraphics[scale=0.3]{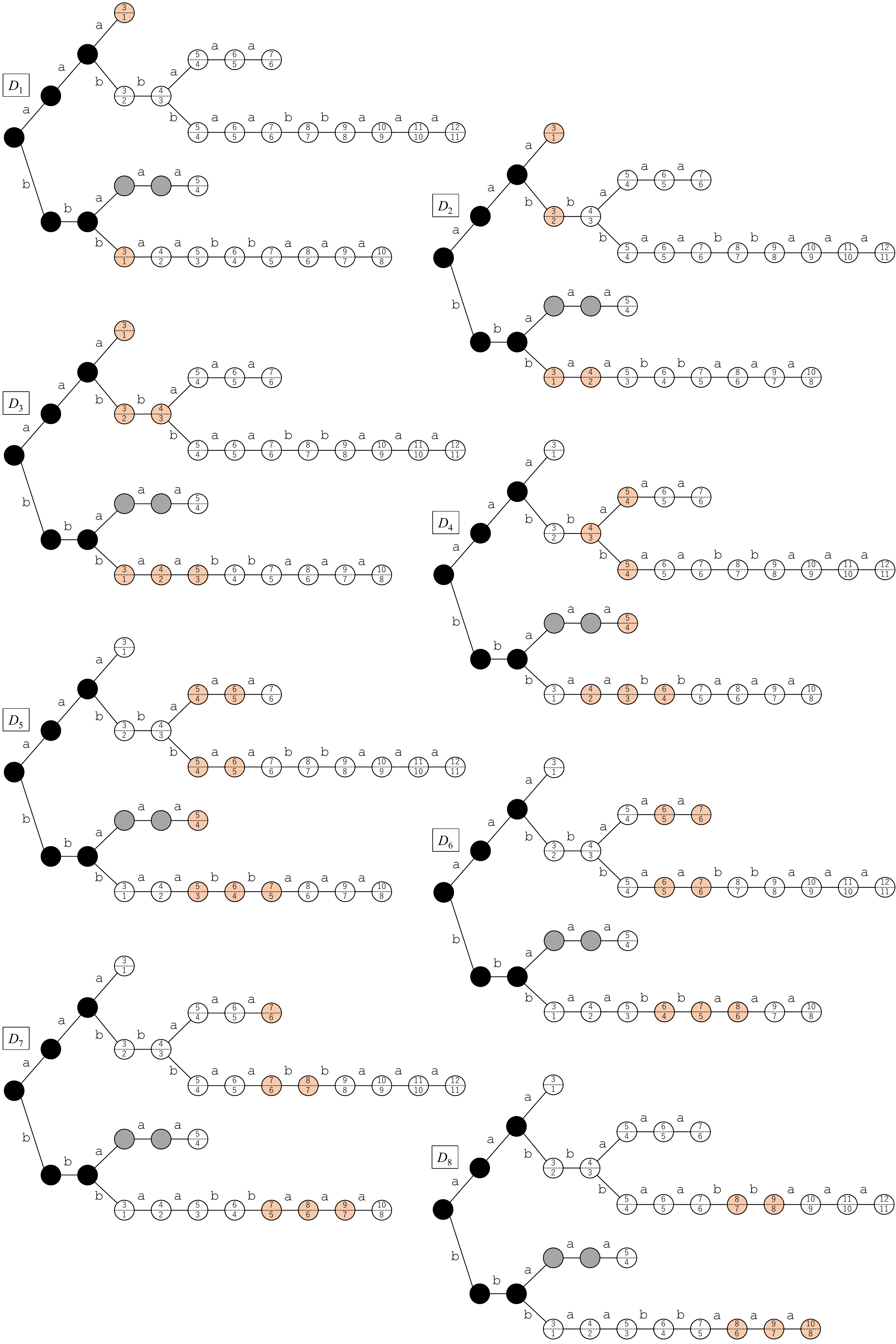}
  \caption{
    Displaying the nodes in $\DMN_{k}$ for $k = 1, 2, \dots, 8$ in orange filled color.
    Two integers in node $v \in \DMN$ represent $\depth(v)$ (upper integer) and $|\strt(v)|$ (lower integer).
    By Lemma~\ref{lem:dmn_for_v}, node $v \in \DMN$ belongs to $\DMN_{k}$ for any $k \in [|\strt(v)|..\depth(v)]$.
  }
\label{fig:rst_k}
\end{center}
\end{figure}

Now we focus on the difference $\ddiff_k = |\DMN_k| - |\DMN_{k-1}|$ between $|\DMN_k|$ and $|\DMN_{k-1}|$, and
show that we can partition $[1..n]$ into $O(r)$ intervals so that $\ddiff_k$ does not change in each interval.
\begin{lemma}\label{lem:diff}
  $|\{ k : \ddiff_{k+1} \neq \ddiff_{k} \}| = O(r)$.
\end{lemma}
\begin{proof}
  We prove the lemma by showing that 
  $|\{ k : |\DMN_{k+1} - \DMN_{k}| \neq |\DMN_k - \DMN_{k-1}| \}| = O(r)$
  and 
  $|\{ k : |\DMN_{k} - \DMN_{k+1}| \neq |\DMN_{k-1} - \DMN_{k}| \}| = O(r)$.

  It follows from Lemma~\ref{lem:dmn_for_v} that a node $v$ in $(\DMN_{k+1} - \DMN_{k})$ satisfies $k+1 = |\strt(v)|$.
  If the parent $u$ of $v$ is in $\DMN$, then $u \in (\DMN_k - \DMN_{k-1})$ since $k = |\strt(u)|$.
  In the opposite direction, any child of $u \in (\DMN_k - \DMN_{k-1})$ is in $(\DMN_{k+1} - \DMN_{k})$ by Lemma~\ref{lem:child_is_dmn}.
  Therefore $|\DMN_{k+1} - \DMN_{k}|$ and $|\DMN_k - \DMN_{k-1}|$ can differ only if one of the following conditions holds:
  \begin{enumerate}
    \item there is a node $v \in (\DMN_{k+1} - \DMN_{k})$ whose parent is not in $\DMN$, which means that $v \in \DMNRoot$;
    \item there is a node $u \in (\DMN_{k} - \DMN_{k-1})$ that has no child or more than one children, which means that $u$ is an explicit node.
  \end{enumerate}
  Note that each node $v$ in $\DMNRoot$ contributes to the first case when $k+1 = |\strt(v)|$ and
  each explicit node $u$ can contribute to the second case when $k = |\strt(u)|$.
  Hence $|\{ k : |\DMN_{k+1} - \DMN_{k}| \neq |\DMN_k - \DMN_{k-1}| \}| = O(r)$.

  It follows from Lemma~\ref{lem:dmn_for_v} that a node $v$ in $(\DMN_{k} - \DMN_{k+1})$ satisfies $k = \depth(v)$.
  If the parent $u$ of $v$ is in $\DMN$, then $u \in (\DMN_{k-1} - \DMN_{k})$ since $k-1 = \depth(u)$.
  In the opposite direction, any child of $u \in (\DMN_{k-1} - \DMN_{k})$ is in $(\DMN_{k} - \DMN_{k+1})$ by Lemmas~\ref{lem:child_is_dmn} and~\ref{lem:dmn_for_v}.
  Therefore $|\DMN_{k} - \DMN_{k+1}|$ and $|\DMN_{k-1} - \DMN_{k}|$ can differ only if one of the following conditions holds:
  \begin{enumerate}
    \item there is a node $v \in (\DMN_{k} - \DMN_{k+1})$ whose parent is not in $\DMN$, which means that $v \in \DMNRoot$;
    \item there is a node $u \in (\DMN_{k-1} - \DMN_{k})$ that has no child or more than one children, which means that $u$ is an explicit node.
  \end{enumerate}
  Note that each node $v$ in $\DMNRoot$ contributes to the first case when $k = \depth(v)$ and
  each explicit node $u$ can contribute to the second case when $k-1 = \depth(u)$.
  Hence $|\{ k : |\DMN_{k} - \DMN_{k+1}| \neq |\DMN_{k-1} - \DMN_{k}| \}| = O(r)$.

  Putting all together, $|\{ k : \ddiff_{k+1} \neq \ddiff_{k} \}| = O(r)$.
\end{proof}

The next lemma will be used in our algorithm presented in Section~\ref{sec:algorithm}.
\begin{lemma}\label{lem:max_in_range}
  Assume that $\ddiff_k = \ddiff$ for any $k \in (k'..k'']$.
  Then, $|\Substr_{\gtext}(k)| / k$, in the range $k \in [k'..k'']$, is maximized at $k'$ or $k''$.
\end{lemma}
\begin{proof}
  For any $k \in [k'..k'']$ with $k = k' + x$, $|\Substr_T(k)| / k$ can be represented by
  \[
    f(x) = \frac{|\Substr_{\gtext}(k')| + \ddiff x }{(k' + x)}. 
  \]
  By differentiating $f(x)$ with respect to $x$, we get
  \begin{align*}
    f'(x) = \frac{\ddiff (k' + x) - (|\Substr_{\gtext}(k')| + \ddiff) k'}{(k' + x)^2} = \frac{\ddiff x - |\Substr_{\gtext}(k')| k'}{(k' + x)^2}.
  \end{align*}
  Assessing sings of $f'(x)$, it turns out that $f(x)$ monotonically decreases for $x \le |\Substr_{\gtext}(k')| k' / \ddiff$ and increases for $x \ge |\Substr_{\gtext}(k')| k' / \ddiff$,
  and hence, $|\Substr_{\gtext}(k)| / k$, in the range $[k'..k'']$, is maximized at $k'$ or $k''$.
\end{proof}

\section{Algorithm}\label{sec:algorithm}
Based on the properties of $\DMN$ established in Section~\ref{sec:properties},
we present an algorithm, given $r$-size run-length compressed string $\gtext$,
to compute $\delta$ in $\ramsort(r, n)$ time and $O(r)$ space.

We first build the r-suffix tree of $\gtext$.
\begin{lemma}\label{lem:rstree}
  Given a string $\gtext$ in run-length compressed form of size $r$,
  the r-suffix tree of $\gtext$ of size $O(r)$ can be computed in $\ramsort(r, n)$ time and $O(r)$ space.
\end{lemma}
\begin{proof}
  Let $w$ be a string of length $r$ that is obtained by replacing each run of $\gtext$ with a meta-character in $[1..r]$,
  where the meta-character of a run $c^e$ is determined by the rank of the run sorted over all runs in $\gtext$ 
  using the sorting key of the pair $(c, e)$ represented in $O(\lg n)$ bits.
  Since $(c, e)$ is represented in $O(\lg n)$ bits, we can compute $w$ in $\ramsort(r, n)$ time and $O(r)$ space.
  Then we build the suffix tree of $w$ in $O(r)$ time and space using any existing linear-time algorithm 
  for building suffix trees over integer alphabets (use e.g.~\cite{Farach1997OST}).
  Since runs with the same character but different exponents have different meta-characters, 
  we may need to merge some prefixes of edges outgoing from a node.
  Fixing this in $O(r)$ time, we get the r-suffix tree of $\gtext$.
\end{proof}

Note that we can reduce $O(\lg n)$-bit characters to $O(\lg r)$-bit characters during the process of sorting in the proof of Lemma~\ref{lem:rstree}.
From now on we assume that a character is represented by an integer in $[1..r]$.
In particular, pointers to some data structures associated with for each character can be easily maintained in $O(r)$ space.

We augment the r-suffix tree in $O(r)$ time and space so that we can support \emph{longest common prefix queries}
that compute $\LCP$ value for any pair of suffixes starting with run's boundary in constant time.
We can implement this with a standard technique that employs lowest common ancsestor queries over suffix trees (see e.g.~\cite{GusfieldBook},
namely, we just compute the string depth of the lowest common ancestor of two leaves corresponding to the suffixes.

Using this augmented r-suffix tree, we can compute the set $\DMNRoot$ of DMN-roots in $O(r)$ time.
\begin{lemma}\label{lem:compute_dmnroot}
  $\DMNRoot$ can be computed in $O(r)$ time.
\end{lemma}
\begin{proof}
  For each leaf $l$, we compute, in the root-to-leaf path, the deepest node $\bnode(l)$ satisfying that 
  there exists $v$ such that $\strt(\bnode(l)) = \strt(v)$ and $\strs(\bnode(l))$ is a proper prefix of $\strs(v)$.
  If $\bnode(l)$ is not a leaf, the child of $\bnode(l)$ (along the path) is the DMN-root on the path.
  Since $\str(\bnode(l))$ and $\str(v)$ are same if we remove their first run, 
  we can compute the longest one by using $\LCP$ queries on the pair of leaves having the same character in the previous run.
  Let $\rfr(l)$ denote the string that can be obtained by removing the first run from $\str(l)$.

  For any character $c$, let $L_{c}$ be the doubly linked list of the leaves starting with the same character $c$ 
  and sorted in the lexicographic order of $\rfr(\cdot)$ (remark that it is not lex.\ order of $\str(\cdot)$).
  Such lists for ``all'' characters $c$ can be computed in a batch in $O(r)$ time by scanning all leaves in the lexicographic order
  and appending a leaf $l$ to $L_{c}$ if the previous character of the suffix $\str(l)$ is $c$.

  Now we focus on the leaves that start with the same character $c$.
  Given $L_{c}$, Algorithm~\ref{algo:dmnroot} computes $|\bnode(\cdot)|$'s in increasing order of the exponents of their first runs.
  When we process a leaf $l$, every leaf with a shorter first run than $l$ is removed from the list so that
  we can efficiently find two lexicographically closest leaves (in terms of $\rfr(\cdot)$ lex.\ order) with longer first runs.
  Let $e$ be the exponent of the first run of $\str(l)$.
  By a linear search to lex.\ smaller (resp.\ larger) direction from $l$ in the current list, 
  we can find lex.\ predecessor $p$ (resp.\ lex.\ successor $s$) that have longer first run than $e$
  (just ignore it if such a leaf does not exist or set sentinels at the both ends of the list).
  Note that the exponent of the first run of $l'$ is $e$ for any leaf $l'$ in between $p$ and $s$.
  Then, we compute $|\bnode(l')| = e + \max (\LCP(\rfr(l'), \rfr(p)), \LCP(\rfr(l'), \rfr(s)))$ and remove $l'$ from the list.
  We can process all the leaves in $L_{c}$ in linear time 
  since any leaf $l$ in $L_{c}$ is visited once and removed from the list after we compute $|\bnode(l)|$ by two $\LCP$ queries.

  After computing $|\bnode(\cdot)|$ for all leaves in $O(r)$ total time, 
  it is easy to locate the nodes of $\DMNRoot$ by traversing the r-suffix tree in $O(r)$ time.
\end{proof}

\begin{algorithm2e}[t]
\caption{How to compute $|\bnode(\cdot)|$ for the leaves in $L_{c}$.}
\label{algo:dmnroot}
\SetKw{True}{true}
\SetKw{Or}{or}
\SetKw{Output}{output}
\SetKw{Return}{return}
\KwIn{Doubly linked list $L_{c}$ of the leaves starting with character $c$ and sorted in the lexicographic order of $\rfr(\cdot)$.}
\KwOut{$|\bnode(\cdot)|$ for the leaves in $L_{c}$.}
$L \leftarrow L_{c}$\tcc*{initialize tentative list $L$ by $L_{c}$}
\ForEach{$l$ in $L_{c}$ in increasing order of the exponents of their first runs}{
  \lIf{$l$ is removed from $L$}{
    continue
  }
  $e \leftarrow$ the exponent of the first run of $l$\;
  compute predecessor $p$ of $l$ (in $\rfr(\cdot)$ lex.\ order) that has longer first run than $e$\;
  compute successor $s$ of $l$ (in $\rfr(\cdot)$ lex.\ order) that has longer first run than $e$\;
  \ForEach{$l'$ existing in between $p$ and $s$ in $L$}{
    \Output $|\bnode(l')| = e + \max (\LCP(\rfr(l'), \rfr(p)), \LCP(\rfr(l'), \rfr(s)))$ for $l'$\;
    remove $l'$ from $L$\;
  }
}
\end{algorithm2e}

We finally come to our main contribution of this paper.
\begin{theorem}
  Given $r$-size run-length compressed string $\gtext$ of length $n$,
  we can compute $\delta$ in $\ramsort(r, n)$ time and $O(r)$ space.
\end{theorem}
\begin{proof}
  Thanks to Lemma~\ref{lem:max_in_range}, in order to compute $\delta$ it suffices to take the maximum of $|\Substr_{\gtext}(k)|/k$ 
  for $k$ at which $\ddiff_{k}$ and $\ddiff_{k+1}$ differ.
  According to the proof of Lemma~\ref{lem:diff}, we obtain $\{ k : \ddiff_{k+1} \neq \ddiff_{k} \}$
  by computing the nodes contributing the changes of $\ddiff_k$, which are the DMN-roots and explicit nodes.
  As shown in Lemmas~\ref{lem:rstree} and~\ref{lem:compute_dmnroot}, DMN-roots can be computed in $\ramsort(r, n)$ time and $O(r)$ space.
  Note that the contributions of $v \in \DMNRoot$ to $\ddiff_{k+1} - \ddiff_k$ are $1$ at $k = |\strt(v)| - 1$ and $-1$ at $k = \depth(v)$,
  and the contributions of an explicit node $u$ in $\DMN$ to $\ddiff_{k+1} - \ddiff_k$  are $h-1$ at $k = |\strt(u)|$ and $1-h$ at $k = \depth(u) + 1$,
  where $h$ is the number of children of $u$.
  We list the information by the pairs $(|\strt(v)| - 1, 1), (\depth(v), -1) \}, (|\strt(u)|, h-1)$ and $(\depth(u) + 1, 1-h)$
  and sort them in increasing order of the first element in $\ramsort(r, n)$ time.
  We obtain $\{ k : \ddiff_{k+1} \neq \ddiff_{k} \}$ by the set of the first elements,
  and we can compute $\ddiff_{k+1} - \ddiff_{k}$ by summing up the second elements for fixed $k$.
  By going through the sorted list, we can keep track of $|\Substr_{\gtext}(k')|$ for $k' \in \{ k : \ddiff_{k+1} \neq \ddiff_{k} \}$,
  and thus, compute $\delta = \max \{ |\Substr_{\gtext}(k')|/k' : k' \in \{ k : \ddiff_{k+1} \neq \ddiff_{k} \} \}$ in $O(r)$ time and space.
\end{proof}

\subparagraph*{Acknowledgements.}
This work was supported by JSPS KAKENHI Grant Number 22K11907.

\clearpage
\bibliography{refs}

\end{document}